\documentclass[letterpaper,english,reprint,nofootinbib,aps,superscriptaddress,showkeys]{revtex4}

\usepackage{babel,calc,amsmath,amsthm,amssymb,graphicx,subfigure,xcolor}
\usepackage{caption}
\usepackage{appendix}
\usepackage{mathrsfs}
\usepackage{bbold}
\usepackage{amsfonts}
\usepackage{dsfont}
\usepackage{float}
\usepackage{tikz}
\usepackage{amsfonts}
\usepackage{rotating}

\usetikzlibrary{positioning, shapes.geometric}
\usepackage[T1]{fontenc}
\usepackage[unicode=true]{hyperref}
\hypersetup{
     colorlinks=true,       		
     linkcolor=blue,          	
     urlcolor=magenta,           	
 }
\usepackage{cleveref}

\newcommand{\Prob}{\mbox{Prob}}

\begin{document}
 \captionsetup[figure]{name={Fig.},labelsep=period}	

\title{Robust violation of a multipartite Bell inequality from the perspective of a single-system game \footnote{Mod. Phys. Lett. A 37, 2250082(2022)}}

\author{Gang-Gang He}
\affiliation{Department of Physics, School of Science, Tianjin University, Tianjin 300072, China}

\author{Xing-Yan Fan}
\affiliation{Theoretical Physics Division, Chern Institute of Mathematics, Nankai University, Tianjin 300071,  China}

\author{ Fu-Lin Zhang}
\email{flzhang@tju.edu.cn }
\affiliation{Department of Physics, School of Science, Tianjin University, Tianjin 300072, China}

\date{\today}
	\begin{abstract}
Recently, Fan \textit{et al.} [Mod. Phys. Lett. A 36, 2150223 (2021)], presented a generalized Clauser-Horne-Shimony-Holt (CHSH) inequality, to identify $N$-qubit Greenberger-Horne-Zeilinger (GHZ) states.
They showed an interesting phenomenon that  the maximal violation of  the generalized CHSH inequality is robust under some specific noises.
In this work, we map the inequality to the CHSH game, and consequently to the CHSH* game in a single-qubit system.
This mapping provides an explanation for the robust violations in $N$-qubit systems.
Namely, the robust violations, resulting from the degeneracy of the generalized CHSH operators
correspond to the symmetry of the maximally entangled two-qubit states and  the identity transformation in the single-qubit game.
This explanation enables us to exactly demonstrate that the degeneracy is $2^{N-2}$.
\end{abstract}


\maketitle
	
\section{Introduction}	
Quantum entanglement \cite{SCat,RevModPhys.80.517,RevModPhys.81.865,PhysRevLett.78.5022,PhysRevLett.80.2245,PhysRevA.100.022320}, brought about by the superpositionprinciple, puzzled many physicists in the early days of quantum theory.
In the original work for Einstein-Podolsky-Rosen (EPR) paradox in entangled states \cite{EPR},
Einstein and his collaborators proposed that quantum mechanics provides probabilistic results because of  its incompleteness.
 In 1964, Bell proposed  an inequality \cite{Bell} to solve the EPR paradox under the assumptions of local reality and hidden variable theory.
Such inequality and its generalized versions \cite{DAS20173928,DAS} revealed nonlocality \cite{RevModPhys.86.419,2019Bell,1998Bell,PhysRevA.100.052121,RevModPhys.81.1727,PhysRevLett.65.1838,PhysRevA.46.5375} in entangled states.
The Clauser-Horne-Shimony-Holt  inequality \cite{PRL.23.880} is the most widely studied Bell inequality for two-qubit systems,
which is written as
 \begin{equation}
 \langle I_{CHSH}^{2}\rangle =\langle A_{0}B_{0}\rangle+\langle A_{0}B_{1}\rangle+\langle A_{1}B_{0}\rangle-\langle A_{1}B_{1}\rangle \leq 2,
 \end{equation}
 with $A_{a},B_{b}$ ($a,b=0,1$) being measurement settings.
 It can be violated by all the two-qubit entangled pure states.

Researchers have tried to understand the nonlocality  from the perspective of game theory \cite{2000Nonlocality,PhysRevA.78.052103,PRA.98.060302(R)}.
The author of Ref. \cite{2000Nonlocality} setted up the so-called CHSH game  to show the  advantage of  quantum strategies.
There are two players, Alice and Bob, in the game who cannot communicate with each other.
They share a two-qubit system and have measurement operators $A_{a}$ and $B_{b}$ respectively.
Here, $a$ and $b$ $=0,1$ are  two input values.
Let $x$ and $y$ $=0,1$ represent the outcomes of Alice and Bob.
When $x\oplus y=ab$, the players win the game, with $\oplus$ denoting modulo $2$ addition.
It is directly to find the linear relationship between their success probability
 \begin{equation}\frac{1}{4}\sum_{a,b} \Prob\left(x\oplus y=ab|a,b\right) \end{equation}
and the  expected value of $\langle I_{CHSH}^{2}\rangle$.
%
%

Recently, Henaut \textit{et al.} \cite{PRA.98.060302(R)} introduced a single-player CHSH* game with two inputs $a$ and $b$.
Any strategy in the CHSH* game can be mapped to a strategy in the CHSH game with two-qubit  maximally entangled states.
Without loss of generality, let Alice and Bob share one of the Bell states, $|\psi_{+}\rangle=\frac{1}{\sqrt{2}}\left(|00\rangle+|11\rangle\right)$.
The player of the CHSH* game, Carol, has a qubit in state $|+\rangle=\frac{1}{\sqrt{2}}\left(|0\rangle+|1\rangle\right)$.
Alice and Bob apply arbitrary local unitary transformations $\mathcal{A}_{a}^{T}$ and $\mathcal{B}_{b}$ to their qubits,
 and then measure the Pauli operator $\sigma_{x}$ on their qubits respectively.
 This is equivalent to  the local measurement of ${A}_{a}{B}_{b}=\mathcal{A}_{a}^{\ast}\sigma_{x}\mathcal{A}_{a}^{T}\otimes \mathcal{B}_{b}^{\dagger}\sigma_{x}\mathcal{B}_{b}$ on $|\psi_{+}\rangle$.
Carol applies $\mathcal{A}_{a}$ and $\mathcal{B}_{b}$ on the state $|+\rangle$ and measures $\sigma_{x}$ on her qubit.
Similarly, this represents the measurement of $C_{ab}=\mathcal{A}_{a}^{\dagger}\mathcal{B}_{b}^{\dagger}\sigma_{x}\mathcal{B}_{b}\mathcal{A}_{a}$ on $|+\rangle$.
She wins the game when her outcome $c = ab \  (\! \!\! \mod 2)$.
The success probabilities of the two games are equal; i.e.,
\begin{equation}
\frac{1}{4}\sum_{a,b} \Prob\left(x\oplus y=ab|a,b\right)= \frac{1}{4} \sum_{a,b} \Prob\left(c=ab|a,b\right),
\end{equation}
which arises from the  expected values
\begin{equation}
\langle I_{CHSH}^{2}\rangle=\langle I_{CHSH}^{1}\rangle,
\end{equation}
with
 $\langle I_{CHSH}^{2}\rangle=\sum_{a,b}(-1)^{ab}\langle\psi_{+}|\mathcal{A}_{a}^{\ast}\sigma_{x}\mathcal{A}_{a}^{T}\otimes \mathcal{B}_{b}^{\dagger}\sigma_{x}\mathcal{B}_{b}|\psi_{+}\rangle$
and $ \langle I_{CHSH}^{1}\rangle=\sum_{a,b}(-1)^{ab}\langle+|\mathcal{A}_{a}^{\dagger}\mathcal{B}_{b}^{\dagger}\sigma_{x}\mathcal{B}_{b}\mathcal{A}_{a}|+\rangle$.

On the other hand,  to distinguish $N$-qubit Greenberger-Horne-Zeilinger (GHZ) states, Fan \textit{et al.} presented a generalized CHSH inequality in their very recent work \cite{2021Greenberger} .
It is expressed as
\begin{equation}\label{CHSHN}
\langle I_{CHSH}^{N}\rangle=\langle \mathbb{A}_{0}\mathbb{B}_{0}\rangle+\langle \mathbb{A}_{0}\mathbb{B}_{1}\rangle+\langle\mathbb{A}_{1}\mathbb{B}_{0}\rangle-\langle\mathbb{A}_{1}\mathbb{B}_{1}\rangle\leq 2.
\end{equation}
$\mathbb{A}_{0}$ and $\mathbb{A}_{1}$ denote the tensor products of local observables for the first $N-1$ qubit.
$\mathbb{B}_{0}$ and $\mathbb{B}_{1}$ represent two observables of the $N$th qubit.
The $N$-qubit GHZ states can be identified by the maximal violations of the inequality.
Besides, they found an interesting quantum phenomenon of robust violations of the generalized CHSH inequality,  in which the maximal violation can be robust under some specific noises.
Such phenomenon originates from the degeneracy of the  largest eigenvalue of Bell-function $I_{CHSH}^{N}$.

In this work, we show the mapping between  the CHSH* and CHSH games proposed by Henaut \emph{et al.} \cite{PRA.98.060302(R)}  can be extended to the generalized CHSH inequality for $N$-qubit case.
The relations among the inequalities (and  the CHSH* game) provide an  explanation for the robust violations
and give the degeneracy of  $I_{CHSH}^{N}$.
Namely, we map the $N$-qubit Bell function $I_{CHSH}^{N}$ to $I_{CHSH}^{2}$ for the two-qubit case, and  consequently to $I_{CHSH}^{1}$ for the one-qubit system.
For the  $N$-qubit GHZ states ($|\psi_+\rangle$ and $|+\rangle$for $N=2$ and $1$), the  expected values satisfy $|\langle I^{1}_{CHSH}\rangle|=|\langle I^{2}_{CHSH}\rangle|\geq|\langle I^{N}_{CHSH}\rangle|$ when $|\langle I^{N}_{CHSH}\rangle|\geq2$.
The equality holds  when the generalized CHSH inequality achieves the Tsirelson's bound,
i.e.
\begin{equation}\langle I^{1}_{CHSH}\rangle = \langle I^{2}_{CHSH}\rangle = \langle I^{N}_{CHSH}\rangle = \pm 2 \sqrt{2}. \end{equation}
This equation is invariable under specific local unitary transformations on the two-qubit and $N$-qubit Bell-functions,
which corresponds to the identity operation on the single-qubit system.
By using these invariance, we exactly demonstrate the degeneracy of $\langle I^{N}_{CHSH}\rangle$, which causes the robust violations.

\section{The generalized CHSH inequality and games}\label{map}

We first introduce the mappings from $\langle I_{CHSH}^{N}\rangle$ to $\langle I_{CHSH}^{2}\rangle$,
where the measured quantum states are the GHZ states
 \begin{equation}
 |G\rangle=\frac{1}{\sqrt{2}}(|00...0\rangle+|11...1\rangle)
 \end{equation}
and the Bell state $|\psi_+\rangle$.
Unless explicitly stated otherwise, all the expected values in this paper are of the states
$|G\rangle$,  $|\psi_+\rangle$ and  $|+\rangle$ corresponding  to $N$-, two- and one-qubit system.
The local measurement operators in $N$-qubit system
can be written as
\begin{equation}\label{eq.2}
X_{j}=\vec{n}_{j}\cdot\vec{\sigma},\quad X_{j}^{\prime}=\vec{n}_{j}^{\prime}\cdot\vec{\sigma}\quad(j=1,2\cdots N),
\end{equation}
where $\vec{\sigma}=(\sigma_{x},\sigma_{y},\sigma_{z})$ is the vector of Pauli matrices, $\vec{n}_{j}=(\sin\alpha_{j}\cos\varphi_{j},\sin\alpha_{j}\sin\varphi_{j},\cos\alpha_{j})$ and $\vec{n}_{j}^{\prime}=(\sin\alpha_{j}^{\prime}\cos\varphi_{j}^{\prime},\sin\alpha_{j}^{\prime}\sin\varphi_{j}^{\prime},\cos\alpha_{j}^{\prime})$ denote the measurement direction of the $j$th qubit.
 Fan \textit{et al.} \cite{2021Greenberger}  defined
the measurement operators of Alice and Bob  in (\ref{CHSHN}) as
\begin{equation}\label{A0}
\mathbb{A}_{0}=\bigotimes_{j=1}^{N-1} X_{j},\mathbb{A}_{1}=\bigotimes_{j=1}^{N-1} X_{j}^{\prime}, \mathbb{B}_{0}=X_{N},\mathbb{B}_{1}=X_{N}^{\prime}.
\end{equation}
 One can derive the first term of $\langle I_{CHSH}^{N}\rangle$  as
 \begin{align}\label{eq.12}
\begin{split}
\langle \mathbb{A}_{0}\otimes\mathbb{B}_{0}\rangle =\frac{1}{2}[1+(-1)^{N}]\prod_{j=1}^{N}\cos\alpha_{j}+\cos(\sum_{j=1}^{N}\varphi_{j})\prod_{j=1}^{N}\sin\alpha_{j}.
\end{split}
\end{align}
 Obviously, only the terms $\prod_{j=1}^{N}\cos\alpha_{j}$, $\cos(\sum_{j=1}^{N}\varphi_{j})$ and $\prod_{j=1}^{N}\sin\alpha_{j}$ are contributed by
  the projections of the measurement operators in the subspace of $\{|00...0\rangle, |11...1\rangle\}$.

For brevity, we ignore the cases:
(i) ${{(\prod_{j=1}^{N-1}\cos\alpha_{j})^{2}+(\prod_{j=1}^{N-1}\sin\alpha_{j})^{2}}}=0$;
(ii) ${{(\prod_{j=1}^{N-1}\cos\alpha_{j}^{\prime})^{2}+(\prod_{j=1}^{N-1}\sin\alpha_{j}^{\prime})^{2}}}=0$;
(iii) $\sin \alpha_N=0$ when $N$ is odd;
(iv) $\sin \alpha_N^{\prime}=0$ when $N$ is odd.
These cases compose a zero measure set, and do not violate the generalized CHSH inequality (\ref{CHSHN}).
To  connect  the expected value to the two-qubit system in the state $|\psi_+\rangle$, we define  the following two mappings.
The first one uniquely leads to a single-qubit observable, for a given $(N-1)$-qubit operator in (\ref{A0}), as
\begin{equation}\label{gamma}
\Gamma[\mathbb{A}_{a}]:=\vec{\mathrm{n}}_{a}\cdot\vec{\sigma}.
\end{equation}
Take $\Gamma[\mathbb{A}_{0}]:=\vec{\mathrm{n}}_{0}\cdot\vec{\sigma}$ as an example.
 The Bloch vector $\vec{\mathrm{n}}_{0}=(\sin\gamma\cos\beta,\sin\gamma\sin\beta,\cos\gamma)$, with
 $\sin\gamma={\prod_{j=1}^{N-1}\sin\alpha_{j}}/{\sqrt{(\prod_{j=1}^{N-1}\cos\alpha_{j})^{2}+(\prod_{j=1}^{N-1}\sin\alpha_{j})^{2}}}$
 and $\beta=\sum_{j=1}^{N-1}\varphi_{j}$.
The other vector $\vec{\mathrm{n}}_{1}$  is in the similar form.
The second mapping projects the single-qubit  observable in (\ref{A0}) onto the equator of Bloch sphere; i.e.,
\begin{equation}\label{theta}
\Theta[\mathbb{B}_{b}]:=\vec{r}_{b}\cdot\vec{\sigma},
\end{equation}
with
$\vec{r}_{0}=(\cos\varphi_{N},\sin\varphi_{N},0)$ and $\vec{r}_{1}=(\cos\varphi_{N}^{\prime},\sin\varphi_{N}^{\prime},0)$.
When $N$ is even,
one can choose
\begin{equation}
A_{0}=\Gamma\left[\mathbb{A}_{0}\right],A_{1}=\Gamma\left[\mathbb{A}_{1}\right],B_{0}=\mathbb{B}_{0},B_{1}=\mathbb{B}_{1},
\end{equation}
while
\begin{equation}\label{B0}
A_{0}=\Gamma\left[\mathbb{A}_{0}\right],A_{1}=\Gamma\left[\mathbb{A}_{1}\right],B_{0}=\Theta[\mathbb{B}_{0}],{B}_{1}=\Theta[\mathbb{B}_{1}]
\end{equation}
when $N$ is odd,
and obtain the two-qubit Bell function
\begin{equation}\label{eq..18}
\langle I_{CHSH}^{2}\rangle=\langle A_{0}B_{0}\rangle+\langle A_{0}B_{1}\rangle+\langle A_{1}B_{0}\rangle-\langle A_{1}B_{1}\rangle,
\end{equation}
corresponding to $\langle I_{CHSH}^{N}\rangle$ in (\ref{CHSHN}).

The above measurement operators $A_{a}$ and $B_{b}$  can always be expressed as ${A}_{a}=\mathcal{A}_{a}^{\ast}\sigma_{x}\mathcal{A}_{a}^{T}$
and ${B}_{b}= \mathcal{B}_{b}^{\dagger}\sigma_{x}\mathcal{B}_{b}$,
with  $\mathcal{A}_{a}^{T}$ and $\mathcal{B}_{b}$  being  two local unitary transformations.
That is, any measurement of the Bell function $\langle I_{CHSH}^{N}\rangle$ in (\ref{CHSHN}) can be mapped to a strategy in the CHSH game,
and consequently to the CHSH* game according to the relation $C_{ab}=\mathcal{A}_{a}^{\dagger}\mathcal{B}_{b}^{\dagger}\sigma_{x}\mathcal{B}_{b}\mathcal{A}_{a}$ given by Henaut \textit{et al.} \cite{PRA.98.060302(R)}.
The two maps are constructed based on the fact that, the expected value of $\mathbb{A}_a\mathbb{B}_b$ on the state $|G\rangle$ is equivalent to the one of $\bar{A}_a\bar{B}_b$ on the state $|\psi_{+}\rangle$. The two single-qubit operators, $\bar{A}_a$ and $\bar{B}_b$, are in the form of (\ref{eq.2}), but the Bloch vector of $\bar{A}_a$ is in the unit ball in general. Therfore, we introduce a normalization coeffieient in (\ref{gamma}).
Under these mappings, we have the two following theorems.

\newtheorem{thm}{\bf Theorem}
\begin{thm}\label{thm1}
$\forall$ $N\geq 3$,
under the mappings in (\ref{gamma}-\ref{eq..18}),  the violation of the CHSH inequality by the Bell state $|\psi_+\rangle$  is  a necessary condition
for  the violation  of the general CHSH inequality by the GHZ state $|G\rangle$.
More particularly,  when $\langle I^{N}_{CHSH}\rangle > 2$,  $\langle I^{2}_{CHSH}\rangle\geq\langle I^{N}_{CHSH}\rangle$;
and when $\langle I^{N}_{CHSH}\rangle < -2$,  $\langle I^{2}_{CHSH}\rangle\leq\langle I^{N}_{CHSH}\rangle$.
\end{thm}
\begin{proof}
 The expected values $\langle I_{CHSH}^{N}\rangle$ and $\langle I_{CHSH}^{2}\rangle$ can be written as
\begin{subequations}
\begin{align}
&\langle I_{CHSH}^{N}\rangle=\sum_{a,b\in\mathbb{Z}_{2}}(-1)^{ab}\langle \mathbb{A}_{a}\otimes \mathbb{B}_{b}\rangle \label{eq..23}\\
&\langle I_{CHSH}^{2}\rangle=\sum_{a,b\in\mathbb{Z}_{2}}(-1)^{ab}\langle  A_{a}\otimes B_{b}\rangle.\label{CHSH2}
\end{align}
\end{subequations}
The  range  of $\langle \mathbb{A}_{a}\otimes \mathbb{B}_{b}\rangle$ is $[-1,1]$.
When $\langle I^{N}_{CHSH}\rangle > 2$, one has
\begin{subequations}\label{eq.13}
\begin{align}
&\sum_{b\in \mathbb{Z}_{2}}\langle \mathbb{A}_{0}\otimes \mathbb{B}_{b}\rangle > 0, \ \ \ \ \ \ \sum_{b\in\mathbb{Z}_{2}}(-1)^{b}\langle \mathbb{A}_{1}\otimes \mathbb{B}_{b}\rangle > 0\\
&\sum_{a\in \mathbb{Z}_{2}}\langle \mathbb{A}_{a}\otimes \mathbb{B}_{0}\rangle > 0, \ \ \ \ \ \ \sum_{a\in\mathbb{Z}_{2}}(-1)^{a}\langle \mathbb{A}_{a}\otimes \mathbb{B}_{1}\rangle > 0.
\end{align}
\end{subequations}
Similarly,  when  $\langle I^{N}_{CHSH}\rangle < -2$,  the greater-than signs in the four inequalities (\ref{eq.13}) become the less-than signs.
Let us denote the normalization constants of $\Gamma[\mathbb{A}_{0}]$ and $\Gamma[\mathbb{A}_{1}]$ as $\varepsilon=\frac{1}{\sqrt{(\prod_{j=1}^{N-1}\cos\alpha_{j})^{2}+(\prod_{j=1}^{N-1}\sin\alpha_{j})^{2}}}$ and $\varepsilon^{\prime}=\frac{1}{\sqrt{(\prod_{j=1}^{N-1}\cos\alpha_{j}^{\prime})^{2}+(\prod_{j=1}^{N-1}\sin\alpha_{j}^{\prime})^{2}}}$.
They satisfy $\varepsilon\geq1$ and $\varepsilon^{\prime}\geq1$.

When $N$ is even,
\begin{align}\label{eq..25}
\begin{split}
&\langle  A_{0}\otimes B_{0}\rangle =\varepsilon\langle \mathbb{A}_{0}\otimes \mathbb{B}_{0}\rangle\\
&\langle  A_{0}\otimes B_{1}\rangle =\varepsilon\langle \mathbb{A}_{0}\otimes \mathbb{B}_{1}\rangle\\
&\langle  A_{1}\otimes B_{0}\rangle =\varepsilon^{\prime}\langle \mathbb{A}_{1}\otimes \mathbb{B}_{0}\rangle\\
&\langle  A_{1}\otimes B_{1}\rangle =\varepsilon^{\prime}\langle \mathbb{A}_{1}\otimes \mathbb{B}_{1}\rangle.
\end{split}
\end{align}
Multiplying by weighting coefficients $(-1)^{ab}$ and summing up them,
according to the  inequalities (\ref{eq.13}a)
one can find
$
\langle I^{2}_{CHSH}\rangle \geq \langle I_{CHSH}^{N}\rangle,
$
when $\langle I^{N}_{CHSH}\rangle > 2$.
Evidenced by the same token,  when $\langle I^{N}_{CHSH}\rangle <-2$, $\langle I^{2}_{CHSH}\rangle \leq \langle I_{CHSH}^{N}\rangle.$

When $N$ is odd, one has
\begin{align}\label{ABOdd}
\begin{split}
&\langle  A_{0}\otimes {B}_{0}\rangle = \frac{\varepsilon}{\sin\alpha_{N}}\langle \mathbb{A}_{0}\otimes \mathbb{B}_{0}\rangle\\
&\langle  A_{0}\otimes {B}_{1}\rangle = \frac{\varepsilon}{\sin\alpha_{N}^{\prime}}\langle \mathbb{A}_{0}\otimes \mathbb{B}_{1}\rangle\\
&\langle  A_{1}\otimes {B}_{0}\rangle = \frac{\varepsilon^{\prime}}{\sin\alpha_{N}}\langle \mathbb{A}_{1}\otimes \mathbb{B}_{0}\rangle\\
&\langle  A_{1}\otimes {B}_{1}\rangle = \frac{\varepsilon^{\prime}}{\sin\alpha_{N}^{\prime}}\langle \mathbb{A}_{1}\otimes \mathbb{B}_{1}\rangle.
\end{split}
\end{align}
When $\langle I^{N}_{CHSH}\rangle > 2$,  according to the  inequalities (\ref{eq.13}b), it is direct to obtain
\begin{equation}\label{eq20}
\frac{1}{\sin\alpha_{N} }  \sum_{a\in \mathbb{Z}_{2}}\langle \mathbb{A}_{a}\otimes \mathbb{B}_{0}\rangle  +   \frac{1}{\sin\alpha_{N}^{\prime} } \sum_{a\in\mathbb{Z}_{2}}(-1)^{a}\langle \mathbb{A}_{a}\otimes \mathbb{B}_{1}\rangle   >\langle I^{N}_{CHSH}\rangle .
\end{equation}
The form of (\ref{eq.12}) leads to  $\frac{\langle \mathbb{A}_{a}\otimes \mathbb{B}_{0}\rangle}{\sin\alpha_{N}} $ and $\frac{\langle\mathbb{A}_{a}\otimes \mathbb{B}_{1}\rangle}{\sin\alpha_{N}^{\prime}}\in[-1,1]$.
Consequently,
\begin{align}\label{eq21}
\begin{split}
&\frac{1}{\sin\alpha_{N} }   \langle \mathbb{A}_{0}\otimes \mathbb{B}_{0}\rangle  +   \frac{1}{\sin\alpha_{N}^{\prime} } \langle \mathbb{A}_{0}\otimes \mathbb{B}_{1}\rangle >0, \\
&\frac{1}{\sin\alpha_{N} }   \langle \mathbb{A}_{1}\otimes \mathbb{B}_{0}\rangle  -   \frac{1}{\sin\alpha_{N}^{\prime} } \langle \mathbb{A}_{1}\otimes \mathbb{B}_{1}\rangle   >0 .
\end{split}
\end{align}
Multiplying by weighting coefficients $(-1)^{ab}$ and summing up the terms in (\ref{ABOdd}), according to the  relations (\ref{eq20}) and (\ref{eq21})
one can find
$
\langle I^{2}_{CHSH}\rangle \geq \langle I_{CHSH}^{N}\rangle,
$
when $\langle I^{N}_{CHSH}\rangle > 2$.
Similarly,  when $\langle I^{N}_{CHSH}\rangle <-2$, $\langle I^{2}_{CHSH}\rangle \leq \langle I_{CHSH}^{N}\rangle.$
This ends the proof.

\end{proof}

There are two corollaries of Theorem \ref{thm1} as follows.
(i) The maximal violations of the GHZ state
cannot exceed the Tsirelson's bound $\pm2\sqrt{2}$, which has been found in Ref. \cite{2021Greenberger}.
(ii) When the expected value for the GHZ state $\langle I^{N}_{CHSH}\rangle= \pm 2\sqrt{2}$ , the corresponding $\langle I^{2}_{CHSH}\rangle= \langle I^{N}_{CHSH}\rangle$.

\begin{thm}\label{thm2}
When the expected value for the GHZ state saturates the Tsirelson's bound, $\langle I^{N}_{CHSH}\rangle= \pm 2\sqrt{2}$,
the Bloch vectors of the operators in $\mathbb{A}_{0}$ and $\mathbb{A}_{1}$  in (\ref{A0})
are restricted as follows three cases
\begin{subequations}
\begin{align}
&\mathrm{(i)}&\   &\vec{n}_{j}=(0,0,\pm1) ,&\vec{n}_{j}^{\prime}&=(\cos\varphi_{j}^{\prime},\sin\varphi_{j}^{\prime},0)\label{eq....2.1}\\
&\mathrm{(ii)}&\ &\vec{n}_{j}=(\cos\varphi_{j},\sin\varphi_{j},0),&\vec{n}_{j}^{\prime}&=(0,0,\pm1)\label{eq....2.2}\\
&\mathrm{(iii)}&\      &\vec{n}_{j}=(\cos\varphi_{j},\sin\varphi_{j},0),&\vec{n}_{j}^{\prime}&=(\cos\varphi_{j}^{\prime},\sin\varphi_{j}^{\prime},0)\label{eq....2.3}.
\end{align}
\end{subequations}
Only the case (iii) exists in the system with an odd $N$.
\end{thm}
\begin{proof}
According to Theorem \ref{thm1},
$\langle I^{N}_{CHSH}\rangle=\langle I^{2}_{CHSH}\rangle$ requires $\varepsilon=\varepsilon^{\prime}=1$.
That is,
${{(\prod_{j=1}^{N-1}\cos\alpha_{j})^{2}+(\prod_{j=1}^{N-1}\sin\alpha_{j})^{2}}}
={{(\prod_{j=1}^{N-1}\cos\alpha_{j}^{\prime})^{2}+(\prod_{j=1}^{N-1}\sin\alpha_{j}^{\prime})^{2}}}
=1$.
Either all of the Bloch vectors $\vec{n}_j$, with $j=1...N-1$, are parallel to z axis, or  perpendicular to z axis.
This holds true for $\vec{n}_j^{\prime}$ with $j=1...N-1$.

When $N$ is odd, $\langle I^{N}_{CHSH}\rangle=\langle I^{2}_{CHSH}\rangle$ also requires $\sin \alpha_N=\sin \alpha_N^{\prime}=1$.
Hence, the Bloch vectors of $B_0$ and $B_1$ in (\ref{B0})  are perpendicular to z axis.
The ones of $A_0$ and $A_1$ should also be perpendicular to z axis, to enable $\langle I^{2}_{CHSH}\rangle$ to achieve $\pm2\sqrt{2}$.
The correspondence  in (\ref{B0})  leads to that only the case (iii)  is allowed.

When $N$ is even, $\vec{n}_j$ and $\vec{n}_j^{\prime}$ cannot be simultaneously parallel to z axis.
This is naturally derived from the fact that the measurement directions of $A_0$ and $A_1$ in (\ref{B0}) are perpendicular when  $\langle I^{2}_{CHSH}\rangle=\pm 2\sqrt{2}$.
In brief, the three cases of the Bloch vectors, (i), (ii) and (iii), are possible when the GHZ state  achieves the maximal violations of the generalized CHSH inequality, with $N$ being even.
This ends the proof.

\end{proof}

\section{Robust Violations of  the generalized CHSH inequality}

\subsection{Framework}

In this section, we show that
the above mappings provide an explanation for the quantum phenomenon of robust violations of the generalized CHSH inequality \cite{2021Greenberger}.
Based on the explanation, one can exactly demonstrate the degeneracy of the Bell function $I_{CHSH}^{N}$,
which corresponds to the dimension of noises for  robust violations.


According to the results in section \ref{map} and Ref. \cite{2021Greenberger},
when the $N$-qubit Bell function $  \langle I_{CHSH}^{N}\rangle = \pm 2\sqrt{2}$,
the corresponding $ \langle I_{CHSH}^{2}\rangle$, and  consequently  $ \langle I_{CHSH}^{1}\rangle$,  reach $\pm 2\sqrt{2}$.
In addition,  the corresponding terms in the three Bell functions are equal; i.e.,
\begin{equation}\label{threeeq}
\langle C_{ab}\rangle =\langle A_{a} B_{b}\rangle= \langle \mathbb{A}_{a} \mathbb{B}_{b}\rangle,
\end{equation}
 with
 \begin{subequations}
\begin{align}
& \langle C_{ab}\rangle =\langle +| \mathcal{A}_{a}^{\dagger}\mathcal{B}_{b}^{\dagger}\sigma_{x}\mathcal{B}_{b}\mathcal{A}_{a} |+\rangle,\\
& \langle A_{a} B_{b}\rangle       =  \langle \psi_+|\mathcal{A}_{a}^{\ast}\sigma_{x}\mathcal{A}_{a}^{T}\otimes \mathcal{B}_{b}^{\dagger}\sigma_{x}\mathcal{B}_{b} |\psi_+ \rangle.
\end{align}
\end{subequations}
One always can inset the $2\times2$ unit operator, $\openone= u^{*} u^{T}$ with  $u$ being unitary, between the two unitary operators in $\langle C_{ab}\rangle $,
as $\langle +| \mathcal{A}_{a}^{\dagger}\mathcal{B}_{b}^{\dagger}\sigma_{x}\mathcal{B}_{b}\mathcal{A}_{a} |+\rangle=
\langle +| \mathcal{A}_{a}^{\dagger}u^{*} u^{T}\mathcal{B}_{b}^{\dagger}\sigma_{x}\mathcal{B}_{b}u^{*} u^{T}\mathcal{A}_{a} |+\rangle$.
It is equivalent to a local unitary transformation on the two qubit system as
\begin{align}\label{u}
\langle \psi_+|\mathcal{A}_{a}^{\ast}\sigma_{x}\mathcal{A}_{a}^{T}\otimes \mathcal{B}_{b}^{\dagger}\sigma_{x}\mathcal{B}_{b} |\psi_+ \rangle
&=\langle \psi_+| (u^{\dag}\otimes u^{T}) (\mathcal{A}_{a}^{\ast}\sigma_{x}\mathcal{A}_{a}^{T}\otimes \mathcal{B}_{b}^{\dagger}\sigma_{x}\mathcal{B}_{b})(u\otimes u^{*}) |\psi_+ \rangle \\
&= \langle \psi_+| (u^{\dag}{A}_{a} u) \otimes  (u^{T} {B}_{b} u^{*}) |\psi_+ \rangle. \nonumber
\end{align}
This actually gives  the symmetry operations of the  Bell state $|\psi_+\rangle$, that
\begin{equation}\label{symm2}
(u\otimes  u^{*})|\psi_{+}\rangle=|\psi_{+}\rangle,
\end{equation}
and the relation between different choices of observers achieving the maximal violation.

To preserve the equations (\ref{threeeq}) and the value of $\langle I_{CHSH}^{N}\rangle$, the local unitary operator $u$ can only have some special forms, which we will list in  the following part.
For a given $u$, the corresponding transformation of the $N$-qubit system can be written as
\begin{equation}\label{uG}
\langle G|   {\mathbb{A}}_{a}  \otimes   {\mathbb{B}}_{b} |G \rangle=\langle G| (\bigotimes_{j=1}^{N}u_{j}^{\dagger}) ( {\mathbb{A}}_{a}  \otimes   {\mathbb{B}}_{b})  (\bigotimes_{j=1}^{N}u_{j})|G \rangle,
\end{equation}
where $u_{N}= u^{*}$.
At this point, note that, ${\mathbb{B}}_{b} =  {B}_{b}$ even if $N$ is odd.
In addition, the set of $u_j$ $(j=1...N-1)$ is not unique.
This is because the mapping from  $\langle I_{CHSH}^{N}\rangle$ to  $\langle I_{CHSH}^{2}\rangle$ is  many-to-one.
Similarly,  the degree of freedom of $u$ in (\ref{u})  also comes from the  many-to-one relationship between $A_a B_b$ and $C_{ab}$.
Generally, $ (\bigotimes_{j=1}^{N}u_{j})|G \rangle$ is different with $|G \rangle$,
 which indicates the largest eigenvalue of $I_{CHSH}^{N}$ is degenerate.
Mixing or superposing $|G\rangle$ with the states in the subspace does not affect $\langle I_{CHSH}^{N}\rangle$.
This is the quantum phenomenon of robust violations of Bell's inequality for the GHZ state presented by Fan  \textit{et al.} \cite{2021Greenberger}.

\begin{figure}[htp]
 \centering
 \includegraphics[width=8cm,height=5.5cm]{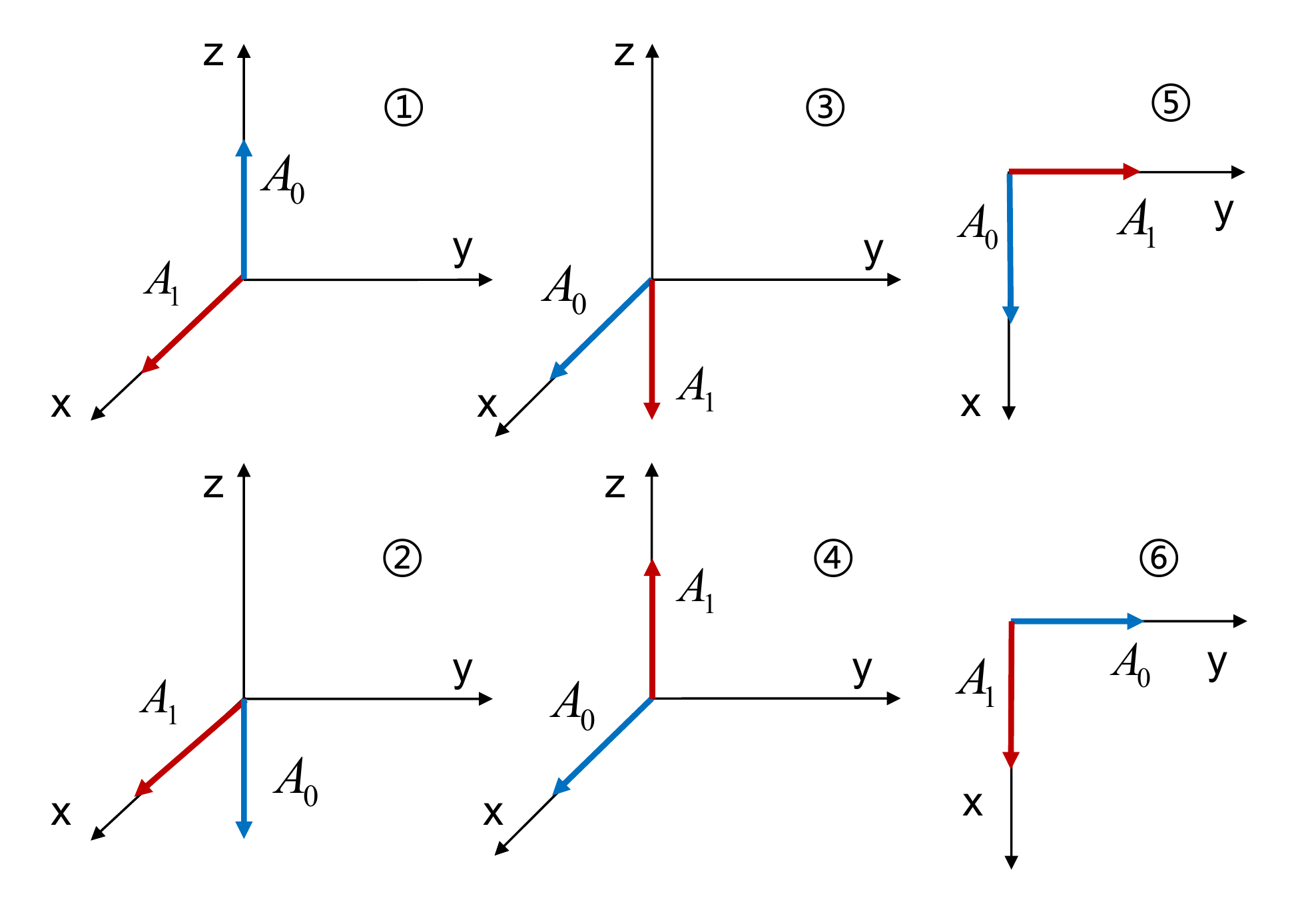}
 \caption{
Relationships between  z axis and the measurement directions of $A_{0}$ and $A_{1}$ in $I_{CHSH}^{2}$, when the corresponding  $\langle I_{CHSH}^{N}\rangle$  reaches the maximal quantum violation.
}\label{fig1}
\end{figure}

\subsection{Degeneracy}

Then, we give the details of the  local unitary transformations and degeneracy.
According to their relationship between the Bloch vectors of and the z axis,
 there are six cases of the observers in $I_{CHSH}^{2}$, as shown in Fig. \ref{fig1}.
Only the measurement directions (up to  rotations about the z axis) of $A_{a}$ are plotted,
 since $B_{a}$ can be uniquely determined by  $A_{a}$ when $\langle I_{CHSH}^{2}\rangle$  reaches the maximal  violation.
These six cases have a  two-to-one correspondences with the three possible choices of the $N$-qubit operators in Theorem \ref{thm2}
which are (i): \textcircled{1}, \textcircled{2}; (ii): \textcircled{3}, \textcircled{4}; and (iii): \textcircled{5}, \textcircled{6}.

According to the equivalence relations under the local unitary transformations on the $N$-qubit system and the exchange between $\mathbb{A}_0$ and $\mathbb{A}_1$,
it is sufficient to consider only the degeneracy of  $ I_{CHSH}^{N}$ corresponding to cases \textcircled{1} and \textcircled{5}.
For the case  \textcircled{1},
 there are six  types of the unitary operator $u$ to consider,  corresponding to the six cases in Fig. \ref{fig1} as the final states  of  $A_{a}$.
However, only  two types of $u$  for the case \textcircled{5}
 need to be considered, corresponding to the final states in \textcircled{5}  and \textcircled{6}.
This is because, it is equivalent to the one in case  \textcircled{1},
 if the $N$-qubit operators $ I_{CHSH}^{N}$ can be transformed by $\bigotimes_{j=1}^{N}u_{j}$ into the cases (i) or (ii)  in Theorem \ref{thm2}.
As show in Fig. \ref{fig2}, for  fixed $\mathbb{A}_a$ and $u$, one can derive the unitary operators $u_{j}$ by requiring $\Gamma(\bigotimes_{j=1}^{N-1}u_{j}^{\dag}  \mathbb{A}_a  \bigotimes_{j=1}^{N-1}u_{j}) = u^{\dag} \Gamma(\mathbb{A}_a) u$.
We remark that,  the initial and final directions of two Bloch vectors can uniquely determine a $2\times 2 $ unitary operator,
up to a phase factor which does not affect  value of $\langle I_{CHSH}^{N}\rangle$ in (\ref{uG}).

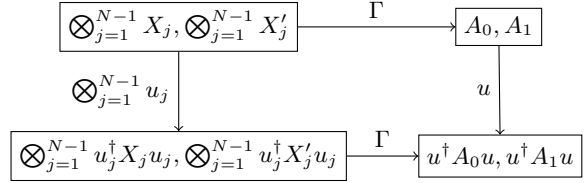
\begin{figure}[htp]
\begin{center}
  \begin{tikzpicture}[node distance=30pt]
  \node[draw, ]                        (start)   {$\bigotimes_{j=1}^{N-1} X_{j},\bigotimes_{j=1}^{N-1} X_{j}^{\prime}$};
  \node[draw, below=of start]                         (step 1) {$\bigotimes_{j=1}^{N-1}u_{j}^{\dagger}X_{j}u_{j},\bigotimes_{j=1}^{N-1}u_{j}^{\dagger}X_{j}^{\prime}u_{j}$};
  \node[draw, right=27pt of step 1]                   (step 2)  {$u^{\dagger}A_{0}u,u^{\dagger}A_{1}u$};
  \node[draw, above=of step 2,right=59.5pt of start]         (choice)  {$A_{0},A_{1}$};

  \draw[->] (start)  -- node[left]  {$\bigotimes_{j=1}^{N-1}u_{j}$}(step 1);
  \draw[->] (step 1) -- node[above] {$\Gamma$}(step 2);
  \draw[->] (start) -- node[above]  {$\Gamma$}(choice);
  \draw[->] (choice) -- node[left]  {$u$}(step 2);

\end{tikzpicture}
\end{center}
\caption{Procedure to  determine the  requirements for $u_j$, $X_j$ and $X_j^{\prime}$.
}\label{fig2}
\end{figure}

\textbf{Case \textcircled{1}.--}
The case  \textcircled{1}  exists only in the system with an even $N$.
The angles $\varphi_j^{\prime}$ in $\mathbb{A}_1$ can always be adjusted to zero by local rotations about the z axis,
which transforms $A_1$ into $\sigma_x$ simultaneously.
In addition, the single-qubit operators in $\mathbb{A}_0$, $X_j=\pm \sigma_z$ have even minus signs.
These minus signs can be removed by qubit flips without affecting the corresponding ${A}_0$.
Therefore, one can choose the initial observables as
\begin{equation}
X_1=...=X_{N-1}= \sigma_z,  \ \ \ X_1^{\prime}=...=X_{N-1}^{\prime}= \sigma_x; \ \ \  {A}_0= \sigma_z, \ \ \   {A}_1= \sigma_x.
\end{equation}

To construct the local unitary transformations $u$ and $u_j$,
we introduce three sets of unitary operators as
\begin{subequations}
\begin{align}
&u^{(1)}=\openone, & &u^{(2)}=\sigma_x,& &u^{(3)}=\exp(i \frac{\sigma_y}{2}\frac{\pi}{2}),& &u^{(4)}=\sigma_z u^{(3)},& &u^{(5)}=\exp(-i \frac{\sigma_x}{2}\frac{\pi}{2}),& &
\!\!\!\!\!\!\!\!\!
u^{(6)}=\sigma_x u^{(5)}; \\
&v^{(1)}=u^{(1)}, & &v^{(2)}=u^{(2)}, & &v^{(3)}=u^{(3)}, & &v^{(4)}=u^{(4)}, &
&v^{(5,6)}=
\begin{cases}
u^{(5,6)}& \text{$N/2$ is even},\\
u^{(6,5)}& \text{$N/2$ is odd};
\end{cases} \!\!\!\!\!\!\!\!\!
 \\
&w^{(1)}=\sigma_x v^{(1)}, & &w^{(2)}=\sigma_x v^{(2)}, & &w^{(3)}=\sigma_z v^{(3)}, & &w^{(4)}=\sigma_z v^{(4)}, & &w^{(5)}=\sigma_x v^{(5)}, & &
\!\!\!\!\!\!\!\!\!
w^{(6)}=\sigma_x v^{(6)}.
\end{align}
\end{subequations}
For an arbitrary superscript  $\nu=1,...,6$, one can check the $N$-qubit GHZ state has a symmetry as
\begin{equation}\label{symmN}
{v^{(\nu)} }^{\otimes N-1} {u^{(\nu)}}^* |G\rangle  = \exp(i\phi)|G\rangle,
\end{equation}
where ${v^{(\nu)} }^{\otimes N-1} $ denotes the direct product of $N-1$ $v^{(\nu)}$ on the first $N-1$ qubits and $\phi\in[0,2\pi]$ is a phase factor.
This can be regarded as an extension of the symmetry of the Bell state in (\ref{symm2}).

The operator $u$, transforming the initial $A_a$ into the case \textcircled{$\nu$},
can be universally written as
\begin{equation}\label{uform}
u= u^{(\nu)} \exp(i\frac{\sigma_z}{2} \delta)
\end{equation}
with $\delta\in[0,2\pi]$.
According to the correspondences between the six cases and  the three possible choices  in Theorem \ref{thm2},
there are two alternative forms of $u_j$ as
\begin{equation}\label{uj}
u_j= v^{(\nu)} \exp(i\frac{\sigma_z}{2} \delta_j) \ \ \ \text{or} \ \ \   w^{(\nu)} \exp(i\frac{\sigma_z}{2} \delta_j),
\end{equation}
with $j=1,...,N-1$ and $\delta_j \in[0,2\pi]$.
Applying them onto $A_a$ and $\mathbb{A}_a$ and requiring $\Gamma(\bigotimes_{j=1}^{N-1}u_{j}^{\dag}  \mathbb{A}_a  \bigotimes_{j=1}^{N-1}u_{j}) = u^{\dag} \Gamma(\mathbb{A}_a) u$,
one can easily obtain the two conditions on $u_{j}$,
as
\begin{equation}\label{delta}
\sum_{j=1}^{N-1} \delta_j =  \delta  \mod 2 \pi,
\end{equation}
 and the number of $ w^{(\nu)}$ in $ \bigotimes_{j=1}^{N-1}u_{j}$ being even.

Then, $\bigotimes_{j=1}^{N}u_{j} |G\rangle$ are eigenstates of the inital $I^{N}_{CHSH}$, with the same eigenvalue as $ |G\rangle $.
By utilizing the forms of $u_{j}$ in (\ref{uj}) and $u_{N}=u^*$ in (\ref{uform}), one can derive these states in three steps :
(1) rotations about the z axis with $\exp(i\frac{\sigma_z}{2} \delta_j)$ and $\exp(-i\frac{\sigma_z}{2} \delta)$;
(2) $v^{(\nu)}$, including the ones in $w^{(\nu)}$,  and $ {u^{(\nu)}}^*$;
(3) the even number of  $\sigma_x$ or $\sigma_z$ factored out from $w^{(\nu)}$.
The state $ |G\rangle $ is invariant under the operations in  the first two steps, because of the condition (\ref{delta}) and  the symmetry (\ref{symmN}).
Consequently, $\bigotimes_{j=1}^{N}u_{j} |G\rangle $ are equivalent to the results of $ |G\rangle $ multiplied by even number of $\sigma_x$ or $\sigma_z$.
Since $ |G\rangle $ is invariant under even number of $\sigma_z$,
 the degenerate states are given by qubit flips in pairs (i.e., application of even number of $\sigma_x$) on $ |G\rangle $.
The degeneracy can be directly derived as
\begin{equation}
   C_{N-1}^{0}+C_{N-1}^{2}+...+C_{N-1}^{N-2} =  2^{N-2}.
\end{equation}

\textbf{Case \textcircled{5}.--}
One can always adjust the angles $\varphi_j$ in $\mathbb{A}_0$  to zero by using local rotations about the z axis,
which transform $A_0$ into $\sigma_x$ and $A_1$ into $\sigma_y$  simultaneously.
Then, the initial observables  can be choose as
\begin{equation}\label{xj5}
X_j= \sigma_x,  \ \ \ X_j^{\prime}= \cos \varphi_j^{\prime} \sigma_x + \sin \varphi_j^{\prime} \sigma_y; \ \ \  {A}_0= \sigma_x, \ \ \   {A}_1= \sigma_y,
\end{equation}
with $j=1,...,N-1$ and $\sum_{j=1}^{N-1} \varphi_j^{\prime} = \pi/2 \mod 2\pi$.

In order to express in a similar way as the case  \textcircled{1}, we define
\begin{subequations}
\begin{align}
&u^{(5)}=\openone, & &v^{(5)}=u^{(5)},& &w^{(5)}= \sigma_x u^{(5)}; \\
&u^{(6)}=\sigma_x, & &v^{(6)}=u^{(6)},& &w^{(6)}= \sigma_x u^{(6)}.
\end{align}
\end{subequations}
Similarly, the $N$-qubit GHZ state has a symmetry as
\begin{equation}\label{symmN5}
{v^{(\mu)} }^{\otimes N-1} {u^{(\mu)}}^* |G\rangle  =\exp(i\theta) |G\rangle,
\end{equation}
with $\mu=5,6$ and $\theta\in [0,2\pi]$ being a phase factor.
The operators transforming the initial $A_a$  and $\mathbb{A}_a$ into the case \textcircled{$\mu$}, can be written as
\begin{eqnarray}\label{uform5}
&&u= u^{(\mu)} \exp(i\frac{\sigma_z}{2} \delta),\\
&&u_j= v^{(\mu)} \exp(i\frac{\sigma_z}{2} \delta_j) \ \ \ \text{or} \ \ \   w^{(\mu)} \exp(i\frac{\sigma_z}{2} \delta_j),
\end{eqnarray}
with $\delta, \delta_j \in[0,2\pi]$ and $j=1,...,N-1$.

We define the sets $J=\{1,2,3,4\cdots,N-1\}$, $K$ and $L$, with $K\subseteq J$ and $L$ being its complementary set.
The elements of $K$ are the subscripts of $u_j$ with $w^{(\mu)}$, and ones of  $L$ are for $v^{(\mu)}$.
Applying the operators (\ref{uform5}) onto $A_a$ and $\mathbb{A}_a$ and requiring $\Gamma(\bigotimes_{j=1}^{N-1}u_{j}^{\dag}  \mathbb{A}_a  \bigotimes_{j=1}^{N-1}u_{j}) = u^{\dag} \Gamma(\mathbb{A}_a) u$,
one can easily obtain
\begin{equation}\label{delta5}
\sum_{j=1}^{N-1} \delta_j =  \delta  \mod 2 \pi,
\end{equation}
 and
 \begin{equation}\label{phi5}
\sum_{k\in K} \varphi_j^{\prime} =  0  \mod \pi.
\end{equation}
The latter condition is on the initial observables $X_{j}^{\prime}$, which is a difference with the case \textcircled{1}.
The states $\bigotimes_{j=1}^{N}u_{j} |G\rangle$ can be derived by following the same three steps in the case \textcircled{1},
which lead to
 \begin{equation}
\bigotimes_{j=1}^{N}u_{j} |G\rangle=\bigotimes_{k\in K} \sigma_x^{k} |G\rangle,
\end{equation}
with $\sigma_x^{k}$ being the Pauli operator $\sigma_x$ of the $k$-th qubit.

For a fixed $K$,  $\bigotimes_{k\in K} \sigma_x^{k} |G\rangle$ reaches the maximal violations, only when the initial $I^N_{CHSH}$ satisfies the condition (\ref{phi5}).
Therefore, the number of $K$, with which the condition (\ref{phi5}) is satisfied, gives  the degeneracy of the largest eigenvalue of $I^N_{CHSH}$.
Then, the maximum degeneracy in the case \textcircled{5} is $2^{N-2}$, which can be obtained based on the following facts.
The condition (\ref{phi5}) cannot  be fulfilled  simultaneously by a subset of $J$ and its  complementary set,
which sets $2^{N-2}$ as the upper limit on the degeneracy.
A simple construction to reach the upper limit is that, $\varphi_1^{\prime}=...=\varphi_{N-2}^{\prime}=0$ and $\varphi_{N-1}^{\prime}=\pi/2$.

\textbf{Example.--}
An arbitrary choice of the operators $\mathbb{A}_a$ and $\mathbb{B}_a$ reaching the maximal violation of  $|G\rangle$
can always be transformed into the above two cases by local unitary operations.
The degenerate subspace can also be derived by the same local unitary operations on the above results.
We show these by using the example with $N=4$ provided in Ref. \cite{2021Greenberger}, which belongs to the case \textcircled{5}.

The parameters of the observables $X_{j}$ and $X_{j}^{\prime}$ are given by $\varphi_{1}=\varphi_{2}=\varphi_{4}^{\prime}=0,\varphi_{1}^{\prime}=\varphi_{2}^{\prime}=\varphi_{4}=\frac{\pi}{2}$, $\varphi_{3}=-\frac{\pi}{4}$ and $\varphi_{3}^{\prime}=\frac{\pi}{4}$.
 Then, $\langle I_{CHSH}^{N}\rangle=\langle I_{CHSH}^{2}\rangle=2\sqrt{2}$.
The operators can be adjusted into the simple form (\ref{xj5}) by using $\tau_3=\exp(i\frac{\sigma_z}{2} \frac{\pi}{4})$  and $\tau_4=\exp(-i\frac{\sigma_z}{2} \frac{\pi}{4})$ on the third and forth qubits. These lead to $\varphi_3 \rightarrow 0$,  $\varphi_4\rightarrow \pi/4$, $\varphi_3^{\prime}\rightarrow \pi/2$ and  $\varphi_4^{\prime}\rightarrow -\pi/4$.
Then, the subsets of $J$, with which the condition (\ref{phi5}) are fulfilled, are given by
\begin{equation}
K=\emptyset, \ \ \ \{1,2\},\ \ \ \{1,3\},\ \ \ \{2,3\}.
\end{equation}
Applying $\tau_3$ and $\tau_4$ onto $\bigotimes_{k\in K} \sigma_x^{k} |G\rangle$,
one obtains the four degenerate states as
\begin{align}
\begin{split}
\frac{1}{\sqrt{2}}(|0000\rangle+|1111\rangle),&  \ \ \ \ \ \  \frac{1}{\sqrt{2}}(|1100\rangle+|0011\rangle)\\
\frac{1}{\sqrt{2}}(e^{-i\frac{\pi}{4}}|1010\rangle+e^{i\frac{\pi}{4}}|0101\rangle),& \ \ \ \ \ \ \frac{1}{\sqrt{2}}(e^{-i\frac{\pi}{4}}|0110\rangle+e^{i\frac{\pi}{4}}|1001\rangle),
\end{split}
\end{align}
which are the same as the results in Ref. \cite{2021Greenberger}.

\section{summary}
In summary, we relate the two recent topics in the area of Bell-nonlocality,
which are the robust violations of Bells inequality of the GHZ states \cite{2021Greenberger} and the single-qubit quantum game \cite{PRA.98.060302(R)}.
Namely, we present the mapping from the generalized  CHSH  inequality, to distinguish the GHZ states constructed by Fan \emph{et al.} \cite{2021Greenberger} to the CHSH game, and consequently to the single-qubit  CHSH* game \cite{PRA.98.060302(R)}.
These relationships provide an explanation for the robust violations of the generalized CHSH inequality in $N$-qubit systems.
The identity transformation in the  CHSH* game,  corresponds to the symmetry of the  two-qubit Bell state,
and  further leads to the local unitary transformations generating the degenerate subspace of the $N$-qubit Bell function.
An arbitrary superposition or mixture in the subspace leads to the same expected value of the Bell function, which is the quantum phenomenon of robust violations.
Based on the explanation, we exactly prove that the maximal degeneracy is $2^{N-2}$.
It would be interesting to extend the mapping among the systems with different numbers of subsystems
to explore more topics in the area of Bell-nonlocality and entanglement, such as the identification of W states.

\begin{acknowledgments}
 This work was supported by the NSF of China (Grants No. 11675119 and No. 11575125).
\end{acknowledgments}

\bibliography{CHSHgame}
\end{document}